\documentclass[11pt,preprint]{elsarticle}

\usepackage{a4wide}
\usepackage[utf8]{inputenc}
\usepackage[T1]{fontenc}

\DeclareUnicodeCharacter{2212}{\ensuremath{-}}
\DeclareUnicodeCharacter{2032}{\ensuremath{'}}

\usepackage[polutonikogreek,english]{babel}
\usepackage{enumerate}
\usepackage{float}
\usepackage{amsmath}
\usepackage{graphicx}
\usepackage{caption}
\usepackage{subcaption}
\usepackage[colorinlistoftodos]{todonotes}
\usepackage{url}
\usepackage{hyperref}
\usepackage{array}
\usepackage{tabularx}
\usepackage{setspace}
\usepackage{abstract}
\usepackage[T1]{fontenc}
\usepackage{array,multirow,makecell}
\setcellgapes{1pt}
\makegapedcells
\newcolumntype{R}[1]{>{\raggedleft\arraybackslash }b{#1}}
\newcolumntype{L}[1]{>{\raggedright\arraybackslash }b{#1}}
\newcolumntype{C}[1]{>{\centering\arraybackslash }b{#1}}

\usepackage{hyperref} 
\hypersetup{
    colorlinks=true,
    linkcolor=blue,
    filecolor=magenta, 
    citebordercolor=red,
    citecolor=red,     
    urlcolor=cyan,
} 
\usepackage{listings}
\usepackage{color}
\definecolor{dkgreen}{rgb}{0,0.6,0}
\definecolor{gray}{rgb}{0.5,0.5,0.5}
\definecolor{mauve}{rgb}{0.58,0,0.82}
\hypersetup{							
pdfauthor = {Sandro GAZZO},			
pdftitle = {Stage CEA - rapport},			
pdfsubject = {Mémoire de Projet},		
pdfkeywords = {Tag1, Tag2, Tag3, ...},	
pdfstartview={FitH}}					

\usepackage{listings}
\usepackage[a4paper]{geometry} 
\usepackage{mathtools,amssymb,amsthm} 
\usepackage{xcolor} 
\usepackage{graphicx} 
  \graphicspath{{images/}} 
\usepackage{float} 
\usepackage{todonotes} 
\usepackage{amsmath} 
\usepackage{caption} 
\usepackage{blkarray,lipsum}

\usepackage{pstricks-add}
\usepackage{pst-eucl}
\usepackage{pict2e} 
\usepackage{array}
\usepackage{multirow}
\usepackage{multicol}
\usepackage{ stmaryrd }
\usepackage{ marvosym }
\usepackage{pifont}
\usepackage{ upgreek }
\usepackage{comment}
\usepackage{ mathrsfs }
\usepackage{systeme}
\usepackage{ dsfont }
\newcommand{\mb}[1]{\mathbf{#1}}
\newcommand{\ms}[1]{\boldsymbol{#1}}
\newcommand{\argmax}[2]{%
\smash{\mathop{{\rm argmax}}\limits_{#1}}\,#2} 

\newtheorem{theorem}{Théorème}
\newtheorem{proposition}[theorem]{Proposition}

\newtheorem*{remark}{Remark}

\theoremstyle{definition}
\newtheorem*{proof_}{Proof}

\usepackage{todonotes} 

\newcommand{\diag}{\operatorname{diag}\diagfences}

\DeclarePairedDelimiter{\diagfences}{(}{)}
\newcommand{\ee}{\mathbb{E}}

\newcommand{\vv}{\mathbb{V}\mathrm{ar}}
\def\rr{{\mathbb R}}
\def\pp{{\mathbb P}}

\title{A generalization of the CIRCE method for quantifying input model uncertainty in presence of several groups of experiments}

\author[1]{Guillaume Damblin\corref{cor1}}
\ead{guillaume.damblin@cea.fr}
\cortext[cor1]{Corresponding author (guillaume.damblin@cea.fr)}
\author[2]{François Bachoc}
\author[3]{Sandro Gazzo
\fnref{fn3}}
\fntext[fn3]{Part of this work was carried out in 2021 when Sandro Gazzo was an intern in CEA.}
\author[4]{Lucia Sargentini}
\author[4]{Alberto Ghione}

\address[1]{Université Paris Saclay - CEA DES/ISAS/DM2S/SGLS/LIAD - F-91191 Gif-sur-Yvette Cedex, France.}
\address[2]{Institut de Mathématiques de Toulouse, UMR5219 CNRS - 31062 Toulouse, France.}
\address[3]{eXalt - 19 rue d'Amiens, 59800 Lille, France.}
\address[4]{Université Paris Saclay - CEA DES/ISAS/DM2S/STMF/LATF - F-91191 Gif-sur-Yvette Cedex, France.}

\begin{document}

\maketitle

\begin{abstract}
The semi-empirical nature of best-estimate models closing the balance equations of thermal-hydraulic (TH) system codes is well-known as a significant source of uncertainty for accuracy of output predictions. This uncertainty, called model uncertainty, is usually represented by multiplicative (log-)Gaussian variables whose estimation requires solving an inverse problem based on a set of adequately chosen real experiments. One method from the TH field, called CIRCE\footnote{This is a French acronym: Calcul des Incertitudes Relatives aux Corrélations \'Elémentaires.}, addresses it. We present in the paper a generalization of this method to several groups of experiments each having their own properties, including different ranges for input conditions and different geometries. An individual (log-)Gaussian distribution is therefore estimated for each group in order to investigate whether the model uncertainty is homogeneous between the groups, or should depend on the group. To this end, a multi-group CIRCE is proposed where a variance parameter is estimated for each group jointly to a mean parameter common to all the groups to preserve the uniqueness of the best-estimate model. The ECME algorithm for Maximum Likelihood Estimation developed in \cite{Celeux10} is extended to the latter context, then applied to relevant demonstration cases. Finally, it is tested on a practical case to assess the uncertainty of critical mass flow assuming two groups due to the difference of geometry between the experimental setups.

\end{abstract}

\section{Introduction}

Numerical simulations are used in the field of nuclear engineering and research to design innovating systems and to achieve safety demonstration and licensing of power plants. In the three past decades, a new methodology has emerged, named Best Estimate Plus Uncertainty (BEPU) \cite{Unal11, Dauria12,Martin2021}, to reproduce nominal and accidental conditions of a nuclear power plant with a thermal-hydraulic system code. This methodology aims to simulate as realistically as possible System Response Quantities (SRQs), for example the peak cladding temperature during a loss-of-coolant accident.


Identifying and quantifying the sources of uncertainty tainting the BEPU simulations is thus a necessary stage as part of a Verification, Validation and Uncertainty Quantification (V\&V UQ) process \cite{Oberkampf10,Rohatgi22}. Those uncertainties mostly arise from numerical approximations (mesh size, nodalization) and physical simplifications due to the use of Closure Relationships (CRs) for mass, momentum and energy balance equations. Being semi-empirical, every CR suffers from model uncertainty, also called epistemic uncertainty. This is often represented by a multiplicative random factor that is intrinsically more suited to thermal-hydraulic experiments with SRQs of different scales (than an additive uncertainty) \cite{Cocci2022}. 



Such a probabilistic treatment of model uncertainty is convenient because of ease of forward propagation through Monte-Carlo sampling \cite{Glaeser08, Peregudov14}. 
The distribution of each random factor is usually estimated by means of some adequate experimental data collected from down-scaled simplified experimental
facilities, known as Separate Effect Tests (SETs) \cite{Aksan08}. A probabilistic equation relates those data to the corresponding simulations whose inputs contain the random factor(s) to estimate. Solving such an inverse problem can be carried out by several methods (see \cite{Wu21} for an up-to-date comprehensive review). Among them, the CIRCE method developed in CEA\footnote{The French Alternative Energies and Atomic Energy Commission (CEA) is a key player in research, development and innovation in nuclear and renewable energies.} is often employed \cite{DeCrecy01,Damblin19}. This method assumes that each factor follows a Gaussian distribution whose mean and variance are estimated by maximizing the
likelihood function related to the inverse problem. The statistical algorithm used by CIRCE for doing so is called ECME for Expectation/Conditional Maximization Either \cite{Celeux10}. 

Among the limitations in using CIRCE, the small
number of experimental data points can lead to poor statistical precision. In such a small data context, the reliability of the estimates also strongly depends on the adequacy of the experimental database with regard to the phenomenon of interest. 
Some guidelines have been presented in the SAPIUM project \cite{Baccou20} to choose adequately the tests of the database by means of the criteria of representativeness and completeness \cite{Baccou20, Baccou18}. Until now in using CIRCE, the adequate experimental database has been made up of several groups of tests reproducing the same physical phenomenon. In general, each group comes from different setups characterized by different geometries and/or input conditions. However, no posterior check has yet been performed to confirm whether or not CIRCE may assume the same magnitude of uncertainty for all the groups. In fact, the resulting uncertainty might be an average of the different probability distributions that would be estimated if CIRCE were applied on each group separately. The risk is therefore to underestimate the uncertainty of one or several group(s). This problem has motivated us to suggest a generalization of the CIRCE method, called the multi-group CIRCE. 

\medskip
Indeed, even though the multiplicative uncertainty factor is generally suitable to capture the different scales of the experimental data, this may not be fully true in some instances. Thus, the probability distribution of the factor should change between the groups in order to fully capture the scaling of the uncertainty. In this context, the multi-group CIRCE aims to evaluate whether or not the multiplicative uncertainty must embed a dependence on the geometry and/or input conditions. 

The multi-group CIRCE is based on a novel ECME algorithm for Maximum likelihood estimation which computes a Gaussian probability distribution for each group. More precisely, a specific variance parameter is estimated for each group but the mean parameter is kept constant across the groups. The latter assumption allows us to preserve the uniqueness of the best-estimate model, that we define as the reference model established by thermal-hydraulic experts shifted by the estimated mean.

\medskip
The paper is organized as follows. Section \ref{sec:CIRCE} recalls the statistical basis of the CIRCE method and guidelines for a proper use. Section \ref{sec:ECME_extension} presents the new ECME algorithm for several groups of experiments. Section \ref{sec:ex} deals with two academic examples that illustrate the multi-group CIRCE. Section \ref{sec:app} deals with a TH application to quantify the uncertainty of critical mass flow. Section \ref{sec:ccl} ends the
paper with conclusions.

\section{The CIRCE method}

\label{sec:CIRCE}

\subsection{Modeling}
The CIRCE method aims to quantify the uncertainty of TH models via multiplicative random factors that are assumed to follow (log-)Gaussian distributions. Let $M_{ref}$ be the reference mathematical expression of a TH model having scalar responses, such that
\begin{equation}
\label{model_uncertainty}
M_{\Lambda}(\mb{x},\theta):=\Lambda\times M_{ref}(\mb{x},\theta)
\end{equation}
with $\mb{x}\in\rr^{d}$ and $\theta\in\rr^{d^{\prime}}$ being respectively input physical conditions (e.g. initial and boundary conditions, geometry parameters) and constant calibration parameters. Since the reference model is multiplied by the (log-)Gaussian random variable $\Lambda$, the resulting model in Eq. (\ref{model_uncertainty}) is random and quantifies model uncertainty. Although there may be experimental and epistemic uncertainties, respectively in $\mb{x}$ \cite{Perrin2019} and $\theta$ \cite{Porter19, Cocci2022}, we assume them to be negligible in this work. 

Models $M_{\Lambda}$ act as closure relationships of the balance equations on which best-estimate system codes rely. Prior to propagation of model uncertainty to output SRQ(s) simulated by those codes, the probability distribution of $\Lambda$ has to be quantified using adequate experimental data collected
from down-scaled simplified facilities called SETs. As experimental data are often only indirect realizations
of $M_{\Lambda}$, the system code, denoted by $G(.)$, will be used to relate them to $\Lambda$. 

As the simulations carried out by $G(.)$ to predict those experimental data may depend on several closure relationships, $\Lambda$ is considered in the sequel as a $p$-dimensional \mbox{(log-)}Gaussian vector $\Lambda\thicksim\mathcal{N}(m,\Sigma)$ with mean vector $m\in\rr^{^p}$ and variance matrix $\Sigma \in M_{p}(\rr)$. Each component $\Lambda_j$ ($1\leq j\leq p$) of $\Lambda$ is thus a \mbox{(log-)}Gaussian variable applied to a particular closure relationship. Then, for $1\leq i\leq n$ with $n$ being the number of selected experimental data, we postulate the following equation relating the latter to the corresponding simulations$:$
\begin{equation}
Y_i=\,\,G(\ms{\lambda}_i,\mb{x}_i)+\epsilon_i
\label{IP}
\end{equation}
where
\begin{itemize}
\item $Y_i\in\rr$ is the $i$th observed experimental SRQ,
\item $\ms{\lambda}_i=(\lambda_{i1},\cdots,\lambda_{ip})^{T}\in\rr^{p}$ is the $i$th unobserved realization of $\Lambda$,
\item $\mb{x}_i$ contains the conditions of the $i$th experiment, assumed known,
\item $\epsilon_i$ is the $i$th observed realization of $E_i\thicksim \mathcal{N}(0,R_i)$ which models the experimental uncertainty tainting $Y_i$. If the variance $R_i$ is unknown, it is set to $0$. In this case, the estimated $\Sigma$ may be inflated to "incorporate" the unknown $R_i$. 
\end{itemize}
Based on Eq. (\ref{IP}), the goal of the CIRCE method is to estimate both the mean vector 
\begin{equation}
m=(m_1,\cdots,m_p)^{T}
\end{equation}
and variance matrix $\Sigma$. As the number $n$ of experimental SRQs is generally small (a few tens/hundreds in the majority of applications), the CIRCE method assumes that the matrix $\Sigma$ is diagonal to reduce the number of parameters being estimated$:$
\begin{equation}
\Sigma=\text{diag}(\sigma_1^2,\cdots,\sigma_p^2).
\end{equation}
This means that the CIRCE method sets that the uncertain factors involved in Eq. (\ref{IP}) are statistically independent among each other.
Furthermore, the method requires that Eq. (\ref{IP}) is linearized in $\ms{\lambda}_i$ at the reference model. Let $H_i=(H_{i1},\cdots,H_{ip})$ be the $p$-dimensional vector consisting, for $1\leq j\leq p$, of the partial derivatives $H_{ij}$ of the code output $G(\ms{\lambda}_i,\mb{x}_i)$ in $\lambda_{ij}$ evaluated at the linearization site $\ms{\lambda}_{nom}=(1,1,\cdots,1)^{T}\in\rr^{p}$ (the $1$-vector of size $p$). Then, Eq. (\ref{IP}) is transformed into
\begin{equation}
\label{eq_circe_linearized}
Y_i^{\prime}=H_i\ms{\lambda}_i^{\prime}+\epsilon_i
\end{equation}
where $Y_i^{\prime}=Y_i-G(\ms{\lambda}_{nom},\mb{x}_{i})$ and $\ms{\lambda}_i^{\prime}=\ms{\lambda}_i-\ms{\lambda}_{nom}\in\rr^{p}$. The vector $Y^{\prime}:=(Y^{\prime}_1,\cdots, Y^{\prime}_n)$ thus follows a multivariate Gaussian distribution with probability density
\begin{equation}
\label{likelihood}
L(Y^{\prime}|m,\Sigma)=\prod_{i=1}^{n}\frac{1}{\sqrt{2\pi(H_i\Sigma H_i^{T}+R_i)}}\exp{\Big(-\frac{1}{2}\frac{(Y_i^{\prime}-H_i(m-\ms{\lambda}_{nom}))^2}{H_i\Sigma H_i^{T}+R_i}\Big)}.
\end{equation}
This function is called the likelihood function. The CIRCE method aims
to maximize it provided that the matrix $H=[H^T_1,\cdots,H^T_n]^{T}\in M_{n,p}(\rr)$ has full rank. 

\medskip
Let us denote a Maximum Likelihood Estimator (MLE) by $\hat{\theta}:=(\hat{m},\hat{\Sigma})\in\rr^{2p}$. The algorithm implemented by CIRCE to compute an MLE is called ECME, for Expectation/Conditional Maximization Either. It can generate successive estimates $\theta^k=(m^k,\Sigma^k)$ which can converge to a MLE when $k$ increases \cite{LiuRubin1994}.

In the rest of the paper, the prime symbol $'$ used in both Eqs. (\ref{eq_circe_linearized}) and (\ref{likelihood}) will be omitted for simplicity. Hence $Y_i$ and $\ms{\lambda}_i$ have to be viewed as the experimental and unobserved model uncertainty realizations shifted respectively by $G(\ms{\lambda}_{nom},\mb{x}_i)$ and $\ms{\lambda}_{nom}$.

\subsection{Estimation via the ECME algorithm}

\label{standard_ecme}

The ECME algorithm relies on the so-called "complete likelihood", which is written as the joint
probability distribution of $Z:=(Y,\lambda)$ with $Y$ and $\lambda$ synthesizing the experimental samples $\{Y_i\}_{1\leq i\leq n}$ and unobserved samples $\{\ms{\lambda}_i\}_{1\leq i\leq n}$ respectively$:$
\begin{equation}
L(Z|m,\Sigma)=L(Y|\lambda,m,\Sigma)L(\lambda|m,\Sigma)
\end{equation}
where both densities are Gaussian, such that
\begin{equation}
L(Y|\lambda,m,\Sigma)=\prod_{i=1}^{n}\frac{1}{\sqrt{2\pi R_i}}\exp{\Big(-\frac{1}{2}\frac{(Y_i-H_i\ms{\lambda}_i)^2}{R_i}\Big)}
\end{equation}
and
\begin{equation}
L(\lambda|m,\Sigma)=\prod_{i=1}^{n}\frac{1}{\sqrt{(2\pi)^{p} |\Sigma|}}\exp{\Big(-\frac{1}{2}(\ms{\lambda}_i-m)^{T}\Sigma^{-1}(\ms{\lambda}_i-m)\Big)}.
\end{equation}
The log-likelihood is then equal to
\begin{equation}
l(Z|m,\Sigma)=l(Y|\lambda,m,\Sigma)+l(\lambda|m,\Sigma).
\end{equation}
After choosing a starting point $(m_0,\Sigma_0)$, the following two steps are run sequentially until
convergence to $\hat{\theta}$.
\begin{enumerate}[1.]
\item \underline{Step of Expectation (E)}: calculation of $$Q((m,\Sigma),(m^k,\Sigma^k))=\ee_{\lambda}[l(Z|m,\Sigma)|Y,m^k,\Sigma^k]$$ where the expectation is taken with respect to the distribution of $\lambda$ conditional on $(Y,m^k,\Sigma^k)$.
\item \underline{Steps of Conditional Maximization (CM)}:
\begin{itemize}

\medskip
\item CM1$:$ $\Sigma^{k+1}=\argmax{\Sigma}{Q((m,\Sigma),(m^k,\Sigma^k))}$,

\smallskip
\item CM2$:$ $m^{k+1}=\argmax{m}{l(Y|m^{k},\Sigma^{k+1})}$.
\end{itemize}
\end{enumerate}
Hereafter, both CM1 and CM2 steps lead to exact updates of $m^{k+1}$ and $\Sigma^{k+1}$ as functions of $m^{k}$ and $\Sigma^{k}$.

\begin{proposition}
\label{ECME}
The two steps of conditional maximization CM1 and CM2 are solved by the following updating formulas. For $1\leq j\leq p$
   \begin{align*}
   CM1. & \quad(\sigma^2)_j^{k+1}=(\sigma^2)_j^k + \frac{1}{n} \sum_{i=1}^n \left[\left(B_{ij}^k (V_i^k)^{-1} A_i^k\right)^2 - (B_{ij}^k)^{2} (V_i^k)^{-1}\right], \\
   CM2. & \quad m^{k+1}= \left( \sum_{i=1}^n H_i^T (V_i^{k+1})^{-1} H_i\right)^{-1} \left( \sum_{i=1}^n H_i^T (V_i^{k+1})^{-1} Y_i \right),
\end{align*}
with scalar quantities $A_i^k$, $B_{ij}^k$ and $V_i^k$ respectively equal to $Y_i-H_im^k$, $(\sigma_j^2)^{k}H_{ij}$ and $H_i\Sigma^{k+1}H_i^{T}+R_i$.
\end{proposition}
\begin{proof}
See \ref{appendixA}. 
\end{proof}

\begin{remark}
Let us stress that the $CM1$ and $CM2$ formulas are special cases of those presented in \cite{Celeux10}. In the latter work, every $Y_i$ can be a vector and $\Sigma$ is not assumed to be a diagonal matrix, which results in updating $\Sigma^{k}$ instead of $\sigma^{k}$.
\end{remark}

The ECME algorithm should be run from several random starting points $(m^0,\Sigma^0)$ to escape possible local maxima. When both $p=1$ and $\epsilon_i=0$ for every $1\leq i\leq n$, the ECME algorithm is unnecessary because a MLE is simply given by the empirical mean and variance of the $Y_iH_i^{-1}$ samples.

\subsection{Statistical reliability of the estimates and falsification}

\label{CIRCEreliability}

\medskip
Let $\hat{\theta}:=(\hat{m},\hat{\Sigma})$ be a MLE. The variance of $\hat{\theta}$ will decrease with increasing $n$ and also depends on the sensitivity of the code responses to the Gaussian variables being estimated. 
\begin{proposition}
The Fisher information of $(m,\Sigma)$ is a block diagonal matrix whose diagonal blocks are respectively equal, for $1\leq j,k\leq p$, to
\begin{equation}
I(m_j,m_k)=\sum_{i=1}^{n}\frac{H_{ij}H_{ik}}{H_i\Sigma H_i^T+R_i}\,\,\,\,\,\,\,\text{and}\,\,\,\,\,\,\,I(\sigma_j^2,\sigma_k^2)=\frac{1}{2}\sum_{i=1}^{n}\frac{H^2_{ij}H^2_{ik}}{(H_i\Sigma H_i^T+R_i)^2}.
\end{equation}
\end{proposition}
\begin{proof_}
See \ref{appendixB}, where the Fisher information matrix is also defined.
\end{proof_}
This proposition can be used to compute an approximation of the asymptotic variance of $\hat{m}$ and $\hat{\sigma}_j^2$ for $1\leq j\leq p$. Indeed, if the cross terms of each block of the Fisher information matrix are negligible, we have
\begin{equation}
\vv[\hat{m}_j]\approx I(m_j,m_j)^{-1}
\,\,\,\,\,\,\,\text{and}\,\,\,\,\,\,\
\vv[\hat{\sigma}^2_j]\approx I(\sigma_j^2,\sigma_j^2)^{-1}.
\end{equation}
Therefore, the higher $n$ or the higher the derivatives components $H_{ij}$ associated with $\Lambda_j$, the smaller the asymptotic variances of $(\hat{m}_j,\hat{\sigma}^2_j)$. An indicator of statistical identifiability of $\Lambda_j$ was introduced in \cite{Celeux10}. It is called NEC (Normalized Error Coefficient), defined as 
\begin{equation}
\label{def_nec}
\text{NEC}_j=\frac{\sqrt{\vv[\hat{m}_j]}}{\hat{\sigma}_j}\,\,\,\,\,\,\,\,1\leq j \leq p.
\end{equation}
The closer this ratio to $0$, the better $\Lambda_j$ is well-identified.

\medskip
In using CIRCE, the normality and linearity assumptions may be flawed. On the one hand, normality must be inspected through the standardized predicted residuals whose probability distribution must be as consistent as possible with that of the standard Gaussian. The standardized residuals, denoted by $e_i$, are equal to$:$
\begin{equation}
\label{residuals}
e_i=\frac{Y_i-H_i\hat{m}}{\sqrt{H_i\hat{\Sigma} H_i^T+R_i}}\,\,\,\,\,\,\,\,1\leq j \leq n.
\end{equation}
A Q-Q plot presenting the theoretical quantiles versus the empirical quantiles of the standardized residuals can
diagnose a possible disagreement to normality. This plot should be complemented by normality tests \cite{Razali11}. On the other hand, the linearity assumption must be inspected across the estimated $95\%$ prediction interval of $\Lambda_j$, equal to
\begin{equation}
\label{IF95_lin}
\mathrm{PI}_{95\%}(\Lambda_j) = [\hat{m}_j-1.96\hat{\sigma}_j,\hat{m}_j+1.96\hat{\sigma}_j]\,\,\,\,\,\,\,\,1\leq j \leq p.
\end{equation}

If $\hat{m}_j$ is not close to $\lambda_{nom,j}$ or $\hat{\sigma}_j$ is too large, the linear approximations and the true code responses can hardly match each other across the whole length of the interval in Eq. (\ref{IF95_lin}). In such a case, the CIRCE method must be re-run after linearizing anew the code responses at $\lambda_{nom,j}:=\hat{m}_j$. This is called iterative CIRCE. Also it is possible that the linear approximations are more accurate in $\log{\Lambda_j}$. The multiplicative factor $\Lambda_j$ thus becomes a log-Gaussian variable and Eq. (\ref{IF95_lin}) is replaced by
\begin{equation}
\label{IF95_exp}
\mathrm{PI}_{95\%}(\Lambda_j) = [\exp{(\hat{m}_j-1.96\hat{\sigma}_j)},\exp{(\hat{m}_j+1.96\hat{\sigma}_j)}].
\end{equation}

\section{Generalization of CIRCE to several groups of experiments}

\label{sec:ECME_extension}

Frequently an adequate experimental database is composed by several data points obtained from different experimental setups. The usual practice is to put all the data points together to form a single experimental
database from which the ECME algorithm is run. The problem is then to mask a possible difference in the magnitude of the uncertainty between the groups. To avoid this, we present an extension of the ECME algorithm where $q$ variance parameters (one for each group) and a single mean parameter (the same for every group) are estimated all together. By contrast with the unpooled model where the mean and variance for each group are estimated separately, the mean parameter is here kept constant between the $q$ groups to preserve the uniqueness of the best-estimate model. 


\subsection{The multi-group ECME algorithm}

From now on, let us denote the whole set of experimental data
\begin{equation}
Y:=(Y^1,\cdots,Y^s,\cdots,Y^q)^{T}\in\rr^{n}\,\,\,\,\,;\,\,\,\,\,1\leq s\leq q
\end{equation}
with $Y^s$ being the vector of experimental data corresponding to the $s$th group of experiments. Thus, $Y$ is now partitioned into $q$ groups $Y^s$ for $1\leq s\leq q$. Let $n_s $ be the size of $Y^s$. Then,
\begin{equation}
n=\sum_{s=1}^{q}n_s.
\end{equation}
We also denote by $i_s$ the index of the last data point of the $s$th group. By setting $i_0:=0$, then $n_s=i_s-i_{s−1}$ for $1\leq s\leq q$. The indexing of data points $Y_i$ composing $Y$ is thus comprised between $1=i_0+1$ and $n=i_q$.

\medskip
Let $\sigma_s^2=(\sigma^2_{s,1},\cdots,\sigma^2_{s,p})^{T}\in\rr^{p}$ be the variance parameter of unobserved realizations $\ms{\lambda}_i$ of the $s$th group for $i_{s-1}+1\leq i \leq i_s$. Thus, we have
\begin{equation}
\ms{\lambda}_i\thicksim\mathcal{N}\big(m,\Sigma_s:=\text{diag}(\sigma_s^2)\big).
\end{equation}
The likelihood can now be written as$:$
\begin{equation}
\label{likelihood_m_groups}
L(Y|m,\Sigma_1,\cdots,\Sigma_q)=\prod_{s=1}^{q}\prod_{i=i_{s-1}+1}^{i_s}\frac{1}{\sqrt{2\pi(H_i\Sigma_s H_i^{T}+R_i)}}\exp{\Big(-\frac{1}{2}\frac{(Y_i-H_im)^2}{H_i\Sigma_s H_i^{T}+R_i}\Big)}.
\end{equation}
The next proposition presents the ECME algorithm extended to this multi-group likelihood.
\begin{proposition}
\label{ECMEbygroups}
The two steps of conditional maximization CM1 and CM2 lead to the following
updating formulas$:$
\begin{align*}
    1)& \quad (\sigma^2_{s,j})^{k+1}=(\sigma^2_{s,j})^{k} + \frac{1}{n_s}\,\sum_{i={i_{s-1}}+1}^{i_s} \left( \left(B_{ij}^k (V_i^k)^{-1}A_i^k\right)^2 - (B_{ij}^k)^2 (V_i^k)^{-1}\right),\\
    & \quad \quad \quad \quad \quad \quad \quad \forall j=1,\dots,p,  \,\,\forall s=1,\dots,q;\\
    2)& \quad m^{k+1}= \left( \sum_{i=1}^n H_i^T (V_i^{k+1})^{-1} H_i\right)^{-1} \left( \sum_{i=1}^n H_i^T (V_i^{k+1})^{-1} Y_i \right).
\end{align*}
where the scalar quantities $A_i^k$, $B_{ij}^k$ and $V_i^k$ are now defined respectively as $Y_i-H_im^k$, $(\sigma_{s,j}^2)^{k}H_{ij}$ and $H_i\Sigma_s^{k}H_i^{T}+R_i$ for $i_{s-1}+1\leq i\leq i_s$.
\end{proposition}
\begin{proof_}
See \ref{appendixC}.
\end{proof_}
Once the algorithm has converged to $\hat{m}$ and $\hat{\sigma}_s^{2}$ ($1\leq s\leq q)$, the identifiability of the estimation for each group must be checked using the NEC indicator presented in Eq. (\ref{def_nec}). In the multi-group context, this indicator is written for the $s$th group of the $j$th factor as 
\begin{equation}
\label{nec_multi_group}
\text{NEC}_{s,j}=\frac{\sqrt{\vv[\hat{m}_j]}}{\hat{\sigma}_{s,j}}\,\,\,\,\,\,\,\,\,1\leq s\leq q,\,\,\,\,\,1\leq j\leq p.
\end{equation}
The components of the block diagonal Fisher information matrix needed to calculate it are now equal to$:$ 
\begin{equation}
\label{fisher_m_multi}
I(m_j,m_k)=\sum_{s=1}^{q}\sum_{i=i_{s-1}+1}^{i_s}\frac{H_{ij}H_{ik}}{H_i\Sigma_s H_i^T+R_i}
\end{equation}
and
\begin{equation}
\label{fisher_var_multi}
I(\sigma^2_{s,j},\sigma^2_{s,k})=\frac{1}{2}\sum_{i=i_{s-1}+1}^{i_s}\frac{H^2_{ij}H^2_{ik}}{(H_i\Sigma_s H_i^T+R_i)^2}.
\end{equation}
The proof is similar to that done for the regular CIRCE.

\subsection{Statistical significance of the variance difference between groups}

A testing procedure can be implemented to
assess the degree of statistical evidence that the variances of the groups are different to one another. To this
end, we can perform a Wald test \cite{Ward2018} applied to the following
null hypothesis$:$
\begin{equation}
\mathcal{H}_0\,:\, \sigma^2_s-\sigma^2_{s^{\prime}}=0,\,\,\,\,1\leq s\neq s^{\prime}\leq q.
\end{equation}
The Wald statistic is written as$:$
\begin{equation}
\label{WaldStat}
W=\frac{(\hat{\sigma}_s^2-\hat{\sigma}^2_{s^{\prime}})^2}{\vv[\hat{\sigma}^2_s]+\vv[\hat{\sigma}^2_{s^{\prime}}]-2\text{Cov}(\hat{\sigma}^2_s,\hat{\sigma}^2_{s^{\prime}})}\thicksim \chi^{2}(1)\,\,\,\text{under}\,\,\,\mathcal{H}_0,
\end{equation}
with $\chi^{2}(1)$ denoting the chi-square distribution with one degree of freedom.
The variance terms in the denominator of Eq. (\ref{WaldStat}) can be obtained from the coefficients of the Fisher information matrix. The test can be applied to each pair of indices $1\leq s\neq s^{\prime}\leq q$.

\medskip
If the test is not rejected at, say, the usual $5\%$ significance level, then it does not necessarily mean that the single variance model is more appropriate than the multi-group model. Another way to decide between the two models is to deal with model selection. In a MLE framework, two models can be compared to each other through the Akaike information criterion (AIC) \cite{AIC1973,AIC2019}, which is a comparison of the likelihood values of each model evaluated at their MLE including a penalty linked to the number of estimated parameters denoted by $n_{\theta}$$:$
\begin{equation}
AIC=2n_{\theta}-2l(Y|\hat{\theta}).
\end{equation}
Here for the regular CIRCE model, $n_{\theta}=2p$, while for the multi-group one, $n_{\theta}=(q+1)p$. The lower the AIC criterion, the better the model quality.

\section{Numerical academic examples}

\label{sec:ex}

\subsection{A first demonstration example}

We present a one-dimensional academic application of the multi-group CIRCE. First, we assume that for $1\leq i\leq n=100$,
\begin{equation}
\label{eqCIRCE_demo}
Y_i=H_i\lambda_i+\epsilon_i
\end{equation}
with
$\lambda_i\thicksim \mathcal{N}(m=1,\sigma^2=0.04)$ and $\epsilon_i\thicksim\mathcal{N}(0,R_i=0.01H_i)$. Hence,
\begin{equation}
Y_i\thicksim \mathcal{N}(H_i,0.04H_i^2+0.01H_i).
\end{equation}
If $H_i$ (and thus $Y_i$) ranges several orders of magnitude, then the distribution of $Y_i$ scales up accordingly (see Figure \ref{circe_reg_vs_circe_2groups} on the left). By contrast, the assumption of the multi-group CIRCE is that several scaling regimes (or groups) exist. An illustration of this with two groups is presented in Figure \ref{circe_reg_vs_circe_2groups} on the right where $\sigma^2$ is tripled from Group 1 to Group 2. Thus we have
\begin{equation}
\label{group1}
Y_i\thicksim \mathcal{N}(H_i,0.04H_i^2+0.01H_i)\,\,\,\,\text{if}\,\,\,\,H_i<10
\end{equation}
\begin{equation}
\label{group2}
Y_i\thicksim \mathcal{N}(H_i,0.12H_i^2+0.01H_i)\,\,\,\,\text{if}\,\,\,\,H_i\geq 10.
\end{equation}

\begin{figure}
\centering
\includegraphics[width=7.2cm]{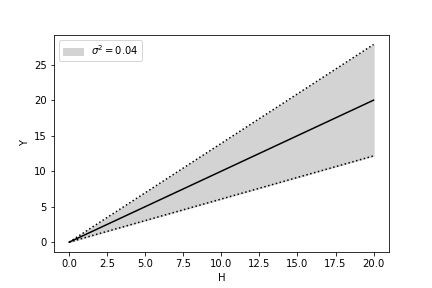}
\includegraphics[width=7.2cm]{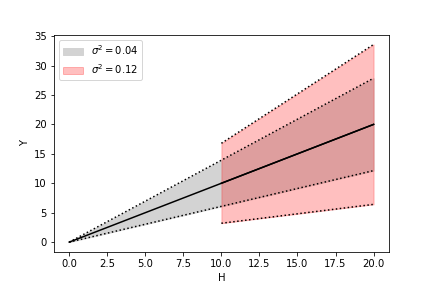}
\caption{Left$:$ Scaling induced by the regular CIRCE $(\sigma^2=0.04)$. Right$:$ Scaling has two regimes$:$ $\sigma^2=0.04$ when $H<10$ and $\sigma^2=0.12$ when $H\geq 10$. The dotted lines delimit the $95\%$ uncertainty.}
\label{circe_reg_vs_circe_2groups}
\end{figure}

\noindent

We have chosen arbitrarily $40$ data points for Group 1 and $60$ data points for Group 2. The two groups are thus made up of $Y_1,\cdots,Y_{40}$ and $Y_{41},\cdots,Y_{100}$ sampled from Eq. (\ref{group1}) and Eq. (\ref{group2}), respectively. The ECME algorithm presented in Section \ref{standard_ecme} was run using all the data points $Y_1,\cdots,Y_{100}$, returning the following MLEs
\begin{equation}
\hat{m}\approx 1.012\,\,\,\,\,\,\,\text{and}\,\,\,\,\,\,\,\hat{\sigma}^2\approx 0.078.
\end{equation}
Figure \ref{circe_reg_vs_circe_2groups_estim} on the left presents the resulting $95\%$ uncertainty interval for $Y$. We can count $0$ and $7$ points outside it for Group $1$ and $2$ respectively. The regular CIRCE thus underestimates the uncertainty of Group 2, and conversely for Group $1$. 

To fix this, the multi-group ECME was run. The MLEs obtained are equal to 
\begin{equation}
\hat{m}\approx 1.002
\end{equation}
and
\begin{equation}
\hat{\sigma}^2_1\approx 0.031\,\,\,\,\,\,\,\, \hat{\sigma}^2_2\approx 0.107.
\end{equation}
These variance estimates are coherent with the true values $(\sigma^2_1,\sigma^2_2)=(0.04,0.12)$. Figure \ref{circe_reg_vs_circe_2groups_estim} on the right presents the corresponding $95\%$ uncertainty interval for $Y$. We can now count $2$ points outside for each group, which is more consistent with the expected $95\%$ coverage frequency. The two scaling regimes have been reconstructed well. The Wald statistic given in Eq. (\ref{WaldStat}) is equal to $W=12.99$ which, as expected, largely exceeds the $5\%$ rejection threshold ($W=3.84$).

\begin{figure}
\centering
\includegraphics[width=7.2cm]{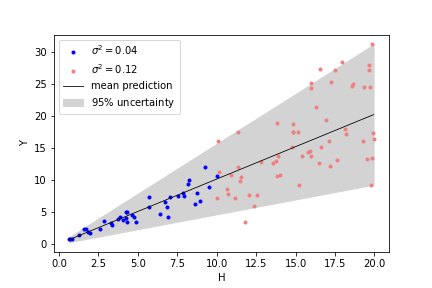}
\includegraphics[width=7.2cm]{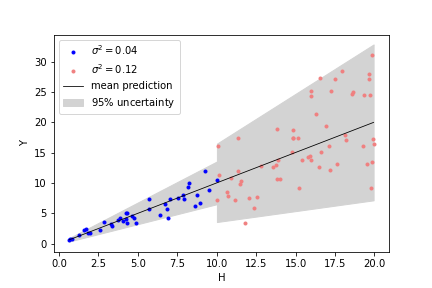}
\caption{Left: estimation by the regular ECME algorithm. Right: estimation by the multi-group ECME algorithm.}
\label{circe_reg_vs_circe_2groups_estim}
\end{figure}


\subsection{An example in higher dimension}

We simulated a vector $Y$ of artificial data according to Eq. (\ref{eqCIRCE_demo}) where $\lambda_i$ is now a $3$-dimensional vector $(p=3)$. We considered three groups with same size $\tilde{n}:=n_1=n_2=n_3$, such that
\begin{equation}
Y=(Y^1,Y^2,Y^3)^{T} \in\rr^{n}, \,\,\,\,n=3\tilde{n}.
\end{equation}
with $Y^s\in\rr^{\tilde{n}}$ ($s=1,2,3$). The data $Y$ were simulated with a matrix $H\in M_{n,3}(\rr)$ whose elements of each column $H_{\star j}$ ($j=1,2,3$) were sampled from the following uniform distributions$:$
\begin{equation}
H_{i 1}\thicksim \mathcal{U}(60,90),\,\,\,\,\,\,
H_{i 2}\thicksim \mathcal{U}(40,70),\,\,\,\,\,\,
H_{i 3}\thicksim \mathcal{U}(20,50),\,\,\,\,\,\,\,1\leq i\leq n.
\end{equation}
The parameters of the three factors for each group were set as follows$:$
\begin{equation}
\label{trus_means_3d}
(m_1,m_2,m_3)=(1,2,4)
\end{equation}
and $\sigma_{s,1}^2=\sigma_{s,2}^2=\sigma_{s,3}^2$
\begin{equation}
\label{true_variances_3d}
=\left\{ 
\begin{array}{ll}
0.9 & s=1 \\
0.3 & s=2 \\
0.6 & s=3 
\end{array}
\right.
\end{equation}
Lastly, the noise variances $R_i$ were all zero.

\medskip
From the simulated data $Y$, both the accuracy and precision of MLE computed by the multi-group ECME algorithm were assessed for increasing sizes of $\tilde{n}$$:$
\begin{equation}
\tilde{n}=125,250,500,1000.
\end{equation}
For each size, $500$ replications were carried out involving the random generation by Eq. (\ref{eqCIRCE_demo}) of $500$ different data sets $Y$. Each replication yielded a MLE $\hat{\theta}$ computed by the multi-group ECME. Let us recall that, as we have considered three factors and three groups for each factor, $\hat{\theta}$ consists of three mean estimates $\hat{m}_j$ ($j=1,2,3$) and three variance estimates $\hat{\sigma}^2_{s,j}$ for each group ($s=1,2,3$). Figures \ref{higher_dimension_violin_mean} and \ref{higher_dimension_violin_sigma2} present violin plots comparing the empirical distributions of the mean and variance estimates respectively. These distributions appear to be close in average to the true values of parameters given by Eqs. (\ref{trus_means_3d}) and (\ref{true_variances_3d}), while the variances computed over the $500$ replications fall with increasing $\tilde{n}$. For each of the three groups ($s=1,2,3$), we can also see in Figure \ref{higher_dimension_violin_sigma2} that the variance component $\sigma^2_{s,j}$ is more precisely estimated for the first factor $j=1$ (followed by the second and third factors $j=2,3$). This was expected because the larger $H_{\star j}$ the better the MLE precision, which is supported by the expression of the Fisher information matrix in the asymptotic regime (see Eqs. (\ref{fisher_m_multi}) and (\ref{fisher_var_multi})).
 
Note that in Figure \ref{higher_dimension_violin_sigma2}, a small portion of the variance estimates are negative. This is a well-known difficulty of mixed effect models in statistics (see for instance \cite{Choi2020} and reference therein). The multi-group CIRCE belonging to this class of models, it is thus affected by this difficulty in this demonstration example. In practice, of course, negative variance estimates should be replaced by zero. More elaborated methods to guarantee non-negative variance estimates exist, again see \cite{Choi2020} and references therein. Nevertheless, these methods are not the object of our work. In particular in the real data considered in Section \ref{sec:app}, all variance estimates are non-negative.

Then, Figure \ref{nec} presents the NEC identifiability indicators given by Eq. (\ref{nec_multi_group}) where $\vv[\hat{m}_j]$ and $\hat{\sigma}_{s,j}$ have been computed over the $500$ replications\footnote{$\hat{\sigma}_{s,j}$ is taken as the expectation of the $500$ corresponding estimates.}. Not surprisingly, every NEC decreases with increasing $\tilde{n}$. Consistently with Figure \ref{higher_dimension_violin_sigma2}, the $\text{NEC}_{s,1}$ indicator is the lowest at fixed $\tilde{n}$ followed by the $\text{NEC}_{s,2}$ and $\text{NEC}_{s,3}$ (for each $s=1,2,3$).  

Although $\tilde{n}=125$ would already be considered as a large database by thermal-hydraulic experts, the associated NEC indicators vary between $0.3$ and $0.83$, which is not indicative of strong identifiability (the closer the NEC to $0$, the stronger the identifiability). This is due to the fact that several factors are estimated together in the inverse problem, which limits the joint statistical identifiability of all the factors involved \cite{Cocci2022b}.

\begin{figure}
\centering
\includegraphics[width=7.2cm]{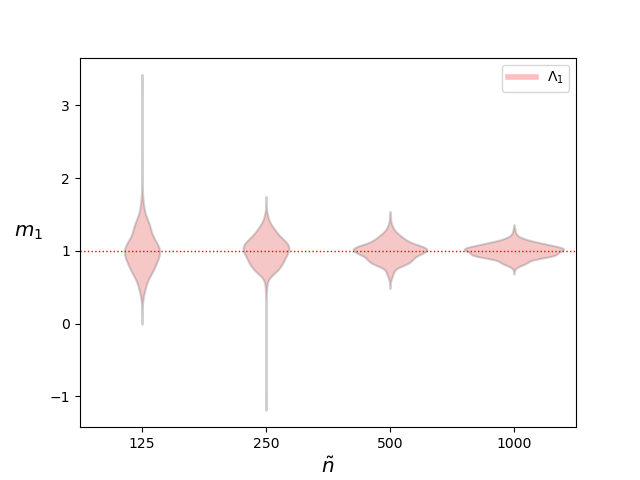}
\includegraphics[width=7.2cm]{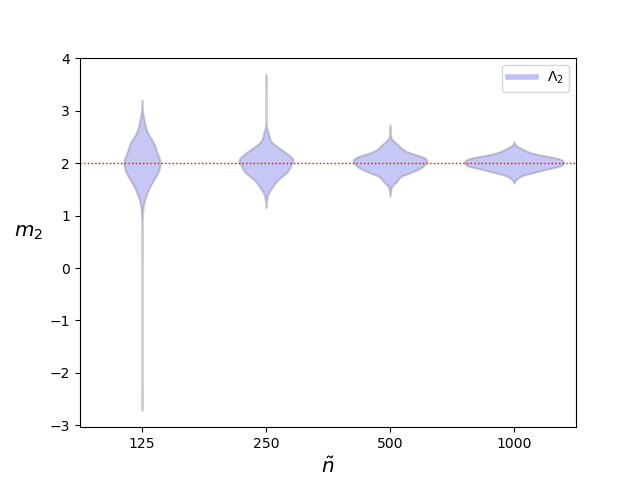}
\includegraphics[width=7.2cm]{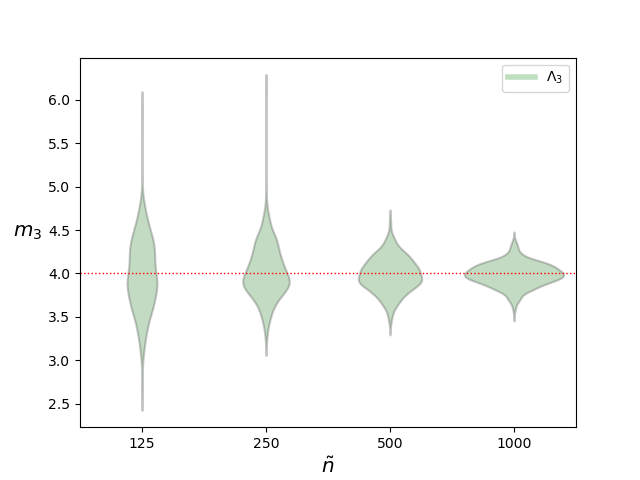}
\caption{Upper left (resp. upper right, bottom)$:$ violin plots of $\hat{m}_1$ (resp. $\hat{m}_2$, $\hat{m}_3$) constructed over $500$ replications. The dotted red line on each plot is the true value of the corresponding mean parameter.}
\label{higher_dimension_violin_mean}
\end{figure}

\begin{figure}
\centering
\includegraphics[width=7.2cm]{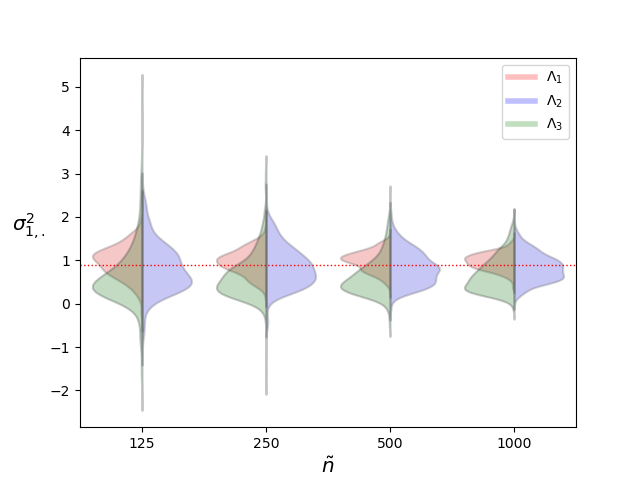}
\includegraphics[width=7.2cm]{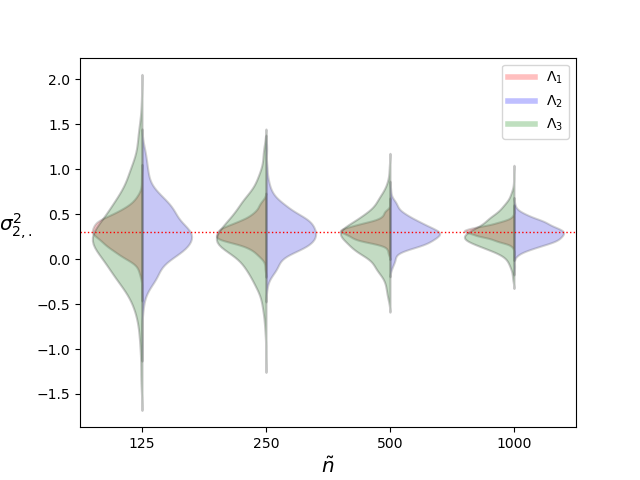}
\includegraphics[width=7.2cm]{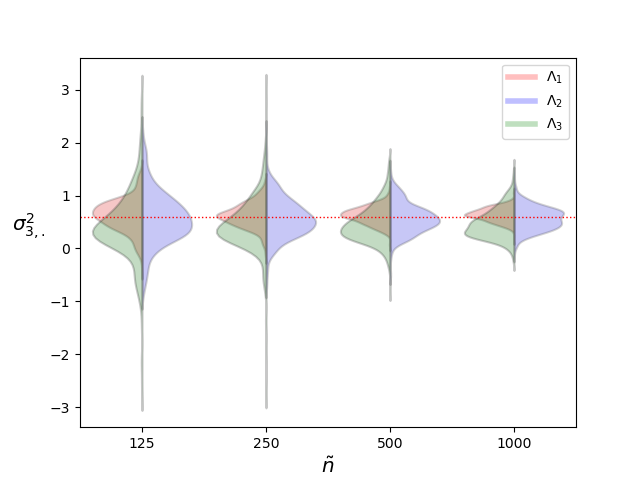}
\caption{Upper left$:$ violin plots of the variance estimates $\hat{\sigma}^2_{1,j}$ of the first group for each factor $\Lambda_j$ constructed over $500$ replications. Upper right (resp. bottom)$:$ Same plot for the variance estimate $\hat{\sigma}^2_{2,j}$ (resp. $\hat{\sigma}^2_{3,j}$). The dotted red line on each plot is the true value of the corresponding variance parameter.}
\label{higher_dimension_violin_sigma2}
\end{figure}

\begin{figure}
\centering
\includegraphics[width=8.2cm]{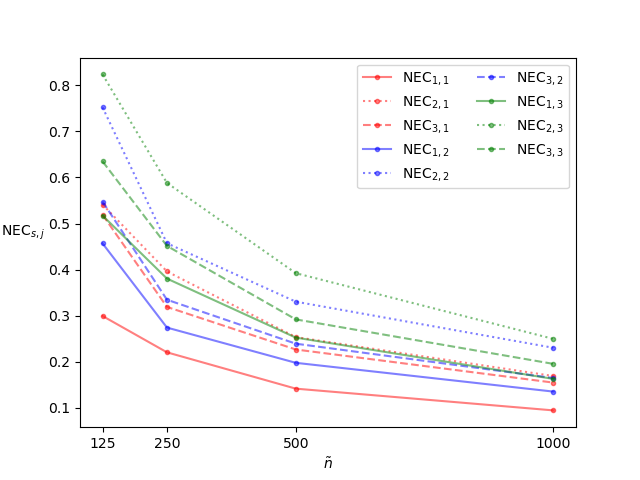}
\caption{Decrease in $\text{NEC}_{s,j}$ on the $y$-axis against the group size $\tilde{n}$ on the $x$-axis.}
\label{nec}
\end{figure}

\medskip
In the next section, we apply the multi-group CIRCE method for uncertainty quantification of the critical mass flow.

\section{Application to the critical mass flowrate}

\label{sec:app}

\subsection{Description of the phenomenon}

The critical mass flow rate is of primary importance to simulate the loss of coolant accident
(LOCA) in a pressurized water reactor (PWR) \cite{Pawluczyk16}. In such an accident, the entire system is depressurized
and the mass inventory in the core decreases because of a discharge of coolant flow
from higher to lower pressure at the break section. As the mass flow rate is independent
of the downstream pressure, its maximum value called critical mass flow (or chocked flow) is
reached. A good prediction of the critical mass flow is thus needed to well-predict the time evolution of the mass inventory in the core.

The study of critical mass flow
is commonly conducted on SETs for which pressures, temperatures and critical mass flow rates can be measured. The experiment of interest comes from the test facility named BETHSY Nozzle, which mimics the break occurring in a LOCA via a convergent test section. Figure \ref{bethsy_facilities} presents the two test sections, that we now refer to B2 and B6 respectively. After a careful inspection of the whole BETHSY database \cite{Sargentini19Nureth}, the adequate experiments picked for uncertainty quantification consist of $n_1=25$ tests from the 2 inches section and $n_2=24$ from the 6 inches section. Table \ref{THranges} presents the physical ranges of input conditions.


\begin{figure}
\centering
\includegraphics[width=9.2cm]{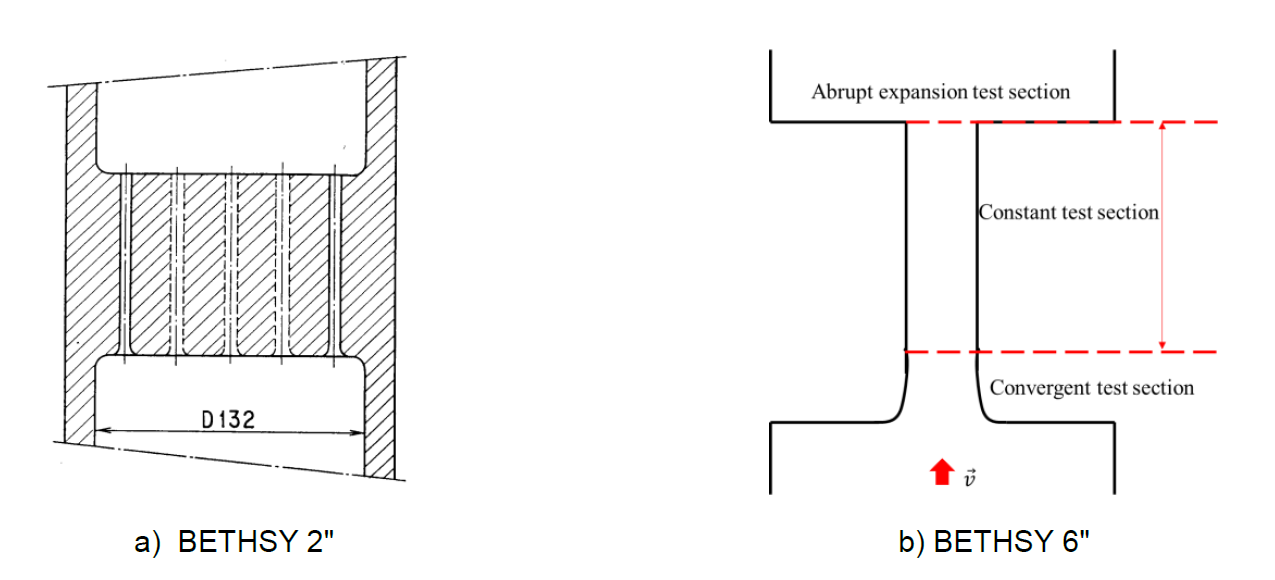}
\caption{Diagram of the two test sections of the BETHSY experiment.}
\label{bethsy_facilities}
\end{figure}

The TH simulations of the B2 and B6 tests are carried out with the CATHARE 2 system code \cite{Geffraye11}. Five closure models of this code may impact the numerical prediction of the critical mass flow rates, including both interfacial and steam-to-liquid friction factors as well as flashing. Recent work has shown that the latter phenomenon is predominant in the code responses accuracy \cite{Sargentini19Nureth}. Thus we neglect the role played by the other
phenomena and then concentrate on uncertainty of the flashing model.

In the next, we will evaluate whether the uncertainty of the flashing model is influenced by the geometry (scaling of the test section). To do so, the results of the regular CIRCE method will be compared to those of the multi-group generalization.

\begin{table}
\centering
\begin{tabular}{|C{4cm}||C{4cm}|}
\hline Pressure (bar) & $30$; $70$; $100$  \\
\hline Temperature ($^{\circ}$C) & $[203, 311]$ \\
\hline
\end{tabular}
\caption{Input conditions ranges of the BETHSY experiments.}
\label{THranges}
\end{table}

\subsection{Application of the regular CIRCE}

Regular CIRCE was first applied with all the $n=49$ mass flow measurements regardless of whether a test is B2 or B6. The best accuracy of the linear approximations of the code responses was obtained in $\log{\Lambda}$. The resulting distribution of $\Lambda$ applied to the flashing model is thus log-Gaussian, parametrized by the mean $m$ and variance $\sigma^2$ of $\log{\Lambda}$. Assuming that the experimental uncertainty is negligible ($\epsilon_i=0$ for $1\leq i\leq 49)$, we have obtained the MLEs\footnote{As $p=1$ and $\epsilon_i=0$, the ECME algorithm is unnecessary to compute them (see the end of Section \ref{standard_ecme}).} below$:$
\begin{equation}
\label{regularCIRCE_results}
\hat{m}=0.68\,\,\,\,\,\,\,\text{and}\,\,\,\,\,\,\,\hat{\sigma}^2=0.21.
\end{equation}
The corresponding $95\%$ prediction interval given in Eq. (\ref{IF95_exp}) is equal to$:$
\begin{equation}
\label{IF95full}
\mathrm{PI}_{95\%}(\Lambda)=[0.81,4.83].
\end{equation}


\begin{figure}
\centering
\includegraphics[width=7cm]{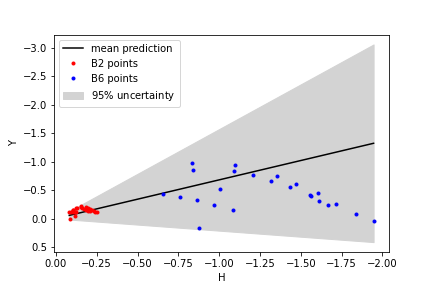}
\includegraphics[width=7cm]{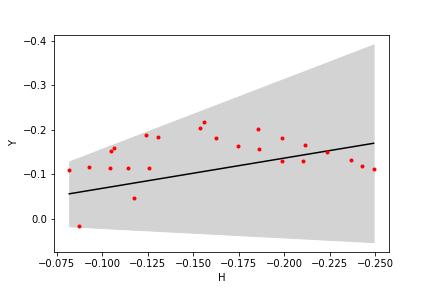}
\caption{Left$:$ Output uncertainty after propagating the regular CIRCE uncertainty. Right$:$ Zoom on the B2 points. The $y$-axis corresponds to the experimental mass flow rates shifted by the reference simulations run at $\lambda_{nom}=1$.}
\label{pred_regular_CIRCE}
\end{figure}
\noindent
The log-normal distribution with parameters given in Eq. (\ref{regularCIRCE_results}) can then be propagated to the model output by the derivative matrix $H$. Figure \ref{pred_regular_CIRCE} displays the resulting $95\%$ uncertainty ranges along with the B2 and B6 group of points. We can observe that the uncertainty is quite oversized for the B6 points colored in blue.
In Figure \ref{standardized_residuals_full}, we have also inspected the normality assumption underpinning CIRCE from the distribution of the standardized predicted residuals. When the residuals include all the B2 and B6 points together, the Kolmogorov-Smirnov (KS) normality test is not rejected at the level of $5\%$ ($p=0.92$). It is rejected if applied only to the B2 points though $(p=0.03)$. Same thing happens if applied only to the B6 points $(p=0.003)$. 

\medskip
The previous analysis shows us that each group of experiments may have its own range of uncertainty, which is overlooked by the regular CIRCE method. Therefore, the multi-group CIRCE with the B2 and B6 groups is carried out in the next section.

\begin{figure}[H]
\centering
\includegraphics[width=7cm]{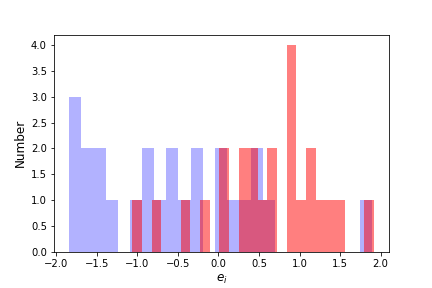}
\caption{Distribution of the standardized predicted residuals $\{e_i\}_{1\leq i\leq 49}$$:$ B2 (red bars) and B6 (blue bars).}
\label{standardized_residuals_full}
\end{figure}

\subsection{Application of the multi-group CIRCE on the B2 and B6 groups}

The multi-group ECME algorithm presented in Section \ref{sec:ECME_extension} has been run to the two groups of B2 and B6 points ($q=2$). The MLEs obtained are equal to
\begin{equation}
\hat{m}=0.57
\end{equation}
and
\begin{equation}
\hat{\sigma}^2_{B2}=0.31\,\,\,\,\,\,\,\, \hat{\sigma}^2_{B6}=0.13.
\end{equation}
The variance of the B2 group is larger than that of the B6 group. The NEC indicators for the first and second group are equal to $0.11$ and $0.17$ respectively. The corresponding $95\%$ prediction intervals are equal respectively to
\begin{equation}
\label{IF95_2p_newECME}
\mathrm{PI}_{95\%}(\Lambda_{B2})=[0.60,5.26]
\end{equation}
and
\begin{equation}
\label{IF95_6p_newECME}
\mathrm{PI}_{95\%}(\Lambda_{B6})=[0.88,3.57].
\end{equation}
The algorithm detects that the magnitude of uncertainty is not the same between the B2 and B6 experiments. The prediction intervals in Eqs. (\ref{IF95_2p_newECME}) and (\ref{IF95_6p_newECME}) are respectively larger and smaller than the one in Eq. (\ref{IF95full}) obtained with the regular CIRCE. Figure \ref{pred_multi_CIRCE} presents the corresponding uncertainty ranges for the critical mass flow rate predictions shifted by the reference calculations, which now scale differently for the two groups. The prediction uncertainty of the B6 group actually gets narrower than that of the regular CIRCE, and conversely for the B2 group.

Figure \ref{standardized_residuals_twogroups} then presents the standardized predicted residuals for each group in red against those coming from the regular CIRCE in blue (which were displayed in Figure \ref{standardized_residuals_full}). We can observe that the distribution of the B2 residuals is a bit less dispersed now, and conversely for the distribution of the B6 residuals. The $p$-values of the KS test for the B2 and B6 residuals are equal to $6\times 10^{-3}$ and $0.053$ respectively. At the level of $5\%$, the Gaussian assumption is still rejected for the B2 experiments, though.

\begin{figure}
\centering
\includegraphics[width=7cm]{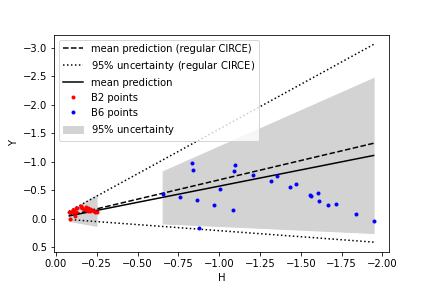}
\includegraphics[width=7cm]{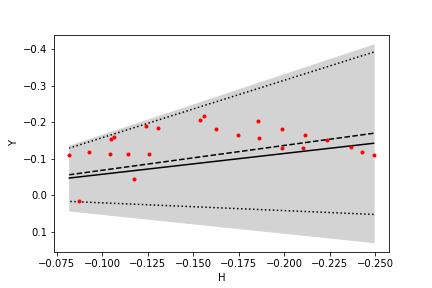}
\caption{Left$:$ Output uncertainty after propagating the multi-group CIRCE uncertainty. Right$:$ Zoom on the B2 points.}
\label{pred_multi_CIRCE}
\end{figure}

Finally, the Wald statistic $W$ presented in Section \ref{sec:ECME_extension} has been calculated to evaluate whether the difference between the two variance estimates is significant. The test relies on the asymptotic variances of $\hat{\sigma}^2_{B2}$ and $\hat{\sigma}^2_{B6}$ equal to the inverse of their Fisher information (see Section \ref{sec:ECME_extension}). We calculated $W=3.62$. As $\pp[\chi^{2}(1)\leq 3.84]=0.95$, the equality of variances is not rejected at the $5\%$ level (although it is close). The AIC criterion is equal to $-19.2$, which is slightly larger than that of the pooled CIRCE $(-19.4)$. From a statistical point of view, there is thus no reason to prefer using the unpooled model to the pooled model until more experimental data are collected to re-run the estimations. Nevertheless, since the uncertainties applied to closure relationships must be demonstrated to be conservative, the unpooled model is less likely to provide undersized uncertainty intervals.

\medskip
In this real application, we do emphasize that the assumption of a constant mean shared by the two groups is likely to be wrong. In fact, most of the B6 predicted residuals are below the mean prediction and conversely for the B2 predicted residuals. However, the multi-group CIRCE can, better than the regular CIRCE, work with the incorrectness of the common mean assumption that may be compensated by inflating the $\sigma^2$ differently for the two groups (which the regular CIRCE cannot do). 


\begin{figure}
\centering
\includegraphics[width=7cm]{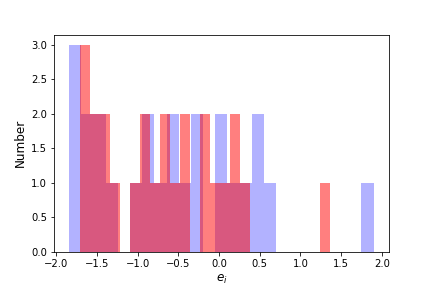}
\includegraphics[width=7cm]{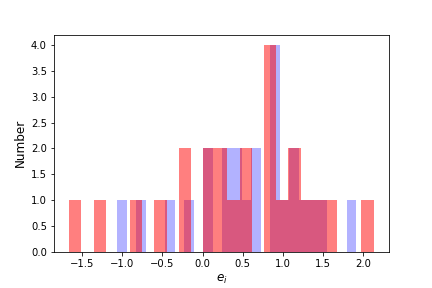}
\caption{Distribution of the standardized predicted residuals$:$ B2 (left) and B6 (right). The blue and red bars relate to regular and multi-group CIRCE, respectively.}
\label{standardized_residuals_twogroups}
\end{figure}

\section{Conclusion}

\label{sec:ccl}

CIRCE is a well-established statistical method based on discrepancies between TH simulations and counterpart real experiments that estimates the probability distributions of the multiplicative factors applied to closure models. The method, which thus falls within the class of Inverse Problems, assumes that the factors are (log-)Gaussian, then runs an ECME algorithm to compute the MLEs of the mean and variance parameters. In practice, most of the closure models integrated into thermal-hydraulic system codes are established from several adequate groups of experiments. The groups may be formed by Separate Effect Tests (SETs) that differ in scale (geometry, TH ranges, etc). Our work has then focused on assessing whether the parameters of the Gaussian factor(s) should depends on the group, or not. 


Therefore, we have developed a multi-group ECME algorithm which estimates a specific variance parameter for each group keeping the mean value across all groups constant. Such per-group algorithm is able to statistically determine either that the magnitude of uncertainty depends on the group of experiments under consideration, or conversely that it is homogeneous between the groups.

A limitation in applying the CIRCE method, even more pronounced with the multi-group ECME algorithm, is when the number of thermal-hydraulic experiments is small, i.e. between a few tens and two hundreds. As this algorithm splits the database into several groups of experiments, the uncertainty of the MLE is likely to increase. As a result, the Wald statistic for testing the equality of the variance parameters can hardly reject the null hypothesis.

Finally, propagating such kind of uncertainties in transient simulations (e.g. a reactor case) remains an open issue. Further investigations will be needed to assess the transition between two different groups and avoid any discontinuity.  

\section{Acknowledgments}
This work has mostly been conducted in 2021 as part of the exploratory project named ACIDITHY (AdequaCy Indices of experimental Database In ThermalHYdraulics), which was funded by the NEEDS program (\url{https://needs.in2p3.fr/}). Furthermore, we would like to greatly thank the two anonymous reviewers who have helped us to improve the paper.

\bibliographystyle{unsrt}
\bibliography{biblio}

\newpage
\appendix

\section{Proof of Proposition \ref{ECME}}

\label{appendixA}

Let us begin by defining $\ms{\delta}_i=(\ms{\lambda}_i-m) \in\rr^{p}$. Then,
\begin{equation}
Y_i=H_i(\ms{\delta}_i+m)+\epsilon_i
\end{equation}
with 
\begin{equation}
\ms{\delta}_i\thicksim\mathcal{N}(0,\Sigma=\diag{\sigma^2}).
\end{equation}
If we denote $\delta_{ij}$ the $j$th component of $\ms{\delta}_i\in\rr^{p}$, then
\begin{equation}
\ee\left( \begin{array}{c}
     Y_i \\
     \delta_{ij} 
\end{array} \right)=\left(\begin{array}{c}
H_{i}m \\
0
\end{array}\right)
\end{equation}
and for $1\leq j\leq p$
\begin{equation}
\vv\left( \begin{array}{c}
     Y_i \\
     \delta_{ij}
\end{array}  \right)=\left(\begin{array}{cc}
H_{i} \diag{\sigma^2} H_{i}^{T}+R & H_{ij} \sigma^2_j \\
\sigma^2_j H_{ij} & \sigma^2_j
\end{array}\right).
\end{equation}
Let us define
\begin{equation}
A_{i}=Y_{i}-H_{i}m ; \quad B_{ij}=\sigma^2_j H_{ij}; \quad V_{i}=H_{i} \Sigma H_{i}^{T}+R_i.
\end{equation}
With these notations, we have
\begin{align}
\nonumber
    \ee\left(\delta_{ij}\mid Y_{i}\right)=&  (\sigma^2_j H_{ij})(H_i \Sigma H_i^T + R_i)^{-1}(Y_i - H_i m) \\
    =& B_{ij} V_{i}^{-1} A_{i}
\label{espcond_delta0}
\end{align}
and
\begin{align}
\nonumber
\vv \left( \delta_{ij}  \mid Y_{i}\right)=& \sigma^2_j - (\sigma^2_j H_{ij}) (H_i \diag{\sigma^2} H_i^T + R_i)^{-1} (H_{ij} \sigma^2_j)\\
=&\sigma^2_j-  B_{ij}^2 V_i^{-1}.
\label{varcond_delta0}
\end{align}
Let $Z_i=(Y_i,\ms{\delta}_i)\in\rr^{p+1}$ be the union of the $i$th experimental response and latent model realization.
The E. step of the ECME algorithm is to derive the expectation of the complete log-likelihood with respect to the distribution of $\delta$ conditional on $(Y,\theta^k)$, equal to
\begin{equation}
\label{Qfunction0}
Q(\theta,\theta^k)=\ee_{\delta}[l(Z|\theta)|Y,\theta^k]
\end{equation}
with $\theta^k=(m^k,(\sigma^2)^{k})\in\rr^{p+1}$ and $\delta$ synthesizing the unobserved samples $\{\ms{\delta}_i\}_{1\leq i\leq n}$. The complete log-likelihood
is written as:
\begin{equation}
l(Z|\theta)=l(Y|\delta,\theta)+l(\delta|\theta).
\end{equation}
The first term is in fact independent of $\theta$ and thus maximizing Eq. (\ref{Qfunction0}) comes down to maximizing $\ee_{\delta}[l(\delta|\theta)|Y,\theta^k]$.
Defining the following canonical parameter $\psi=(\sigma^2)^{-1}\in\rr^{p}$, $l(\delta|\psi)$ can be written as
\begin{equation}
\label{logQpartial0}
l(\delta|\psi)=-\frac{np}{2}\ln{(2\pi)}+\frac{n}{2}\ln{|\diag{\psi}|} -\frac{1}{2}\sum_{i=1}^{n}  (\ms{\delta}_{i}^2)^{T}\psi.
\end{equation}
We can rewrite Eq. (\ref{logQpartial0}) in the following way
\begin{equation}
l(\delta|\psi)=\psi^T t(\delta) - \ln (a(\psi)) + \ln (b(\delta)) 
\end{equation}
with
\begin{itemize}

\item $t(\delta)=-1/2\sum_{i=1}^{n}  \ms{\delta}_i^2$ the sufficient statistics of $\lambda$;

\smallskip
\item $\ln(a(\psi))= -\frac{n}{2} \ln|\diag{\psi}|   $;

\smallskip
\item $\ln(b(\delta))=-\frac{np}{2}\ln(2\pi)$ depending on the length of $\delta$. 
\end{itemize} 

\medskip
Considering the canonical parameter instead of $\theta$, we have (up to a constant)
\begin{align}
\nonumber
Q(\psi,\psi^k)= &\,\ee_{\delta}[l(\delta|\psi)|Y,\psi^k]\\
\nonumber
=& \,\ee_{\delta}(\psi^T t(\delta)|Y,\psi^k) - \ee_{\delta}(\ln(a(\psi))|Y,\psi^k) + \ee_{\delta}(\ln(b(\delta))|Y,\psi^k) \\
    =& \,\psi^T \ee_{\delta}(t(\delta)|Y,\psi^k) - \ln(a(\psi)) + \ln(b(\delta))
\label{Estep0}
\end{align}
where the $j$th component of $\ee_{\delta}(t(\delta)|Y,\psi^k)$ is equal to
\begin{equation}
 \ee\Big(-\frac{1}{2} \sum_{i=1}^{n} \delta^2_{ij} \mid Y, \theta^{k}\Big)=-\frac{1}{2} \sum_{i=1}^{n}\left[\ee\left(\delta_{ij} \mid Y, \theta^{k}\right)^2 +\vv\left(\delta_{ij} \mid Y, \theta^{k}\right)\right].
\end{equation} 
Using Eqs. (\ref{espcond_delta0}) and (\ref{varcond_delta0}) we thus have
\begin{multline}
\ee\Big(-\frac{1}{2} \sum_{i=1}^{n} \delta^2_{ij} \mid Y, \theta^{k}\Big) = \\
-\frac{1}{2}\Big(n (\sigma^2_j)^k + \sum_{i=1}^{n} \Big[(B_{ij}^k(V_i^k)^{-1} A_i^k)^2  - (B_{ij}^k)^2(V_i^k)^{-1}\Big]\Big).
\end{multline}
Then, the CM1 step is to find $\psi^{k+1}$ that maximizes (\ref{Estep0}).
\begin{align*}
    &\frac{\partial Q(\psi,\psi^k)}{\partial \psi}=0\\
    \Longleftrightarrow &\,\,\,\frac{\partial \ee_{\delta}[l(\delta|\psi)|Y,\psi^k])}{\partial \psi}=0\\
    \Longleftrightarrow &\,\,\, \frac{\partial [\psi^T \ee_{\delta}(t(\delta)|Y,\psi^k) - \ln(a(\psi)) + \ln(b(\delta))]}{\partial \psi}=0\\
    \Longleftrightarrow & \,\,\,\ee_{\delta}(t(\delta)|Y,\psi^k) - \frac{\partial\ln(a(\psi))}{\partial \psi}=0\\
    \Longleftrightarrow &\,\,\, \ee_{\delta}(t(\delta)|Y,\psi^k) = \frac{\partial \ln(a)}{\partial \psi}(\psi).
\end{align*}
It comes out that $\psi_{k+1}$ achieves the following equation
\begin{equation}
\ee_{\delta}(t(\delta)|Y,\psi^k) = \frac{\partial \ln(a)}{\partial \psi}(\psi_{k+1})
\end{equation}
where
\begin{align}
    \frac{\partial \ln(a)}{\partial \psi}(\psi)=&-\frac{n}{2} \psi^{-1}\\
    =& -\frac{n}{2} \sigma^2.
\end{align}
Finally, we can deduce the expression of each component of $\sigma^2=\psi^{-1}$ at the $k+1$th iteration of this ECME algorithm$:$
\begin{equation}
(\sigma_j^2)^{k+1}=(\sigma_j^2)^k + \frac{1}{n}\sum_{i=1}^{n} \left( \left(B_{ij}^k (V_i^k)^{-1}A_i^k\right)^2 - (B_{ij}^k)^2 (V_i^k)^{-1} \right).
\end{equation}
The CM2 step is then to compute $m_{k+1}$ that maximizes the incomplete likelihood conditional on $(\sigma_j^2)^{k+1}$. It leads to the weighted least squares estimator below
\begin{equation}
m^{k+1}= \left( \sum_{i=1}^n H_i^T (V_i^{k+1})^{-1} H_i\right)^{-1} \left( \sum_{i=1}^n H_i^T (V_i^{k+1})^{-1} Y_i \right).
\end{equation}

\section{The Fisher information matrix}

\label{appendixB}

The calculation of the Fisher information matrix, called here $I$, allows us to derive the asymptotic variances of $\hat{\theta}=(\hat{m},\hat{\sigma}^2)$. For $1\leq i,j\leq 2p$, then
\begin{equation}
\label{Fisher_def}
I(\theta_i,\theta_j)=-\ee_{Y}\Big[\frac{\partial^2}{\partial\theta_i\partial \theta_j}l(Y|\theta)|\theta\Big]
\end{equation}
where
\begin{equation}
l(Y|\theta)=-\frac{n}{2}\ln{(2\pi)}-\frac{1}{2}\sum_{i=1}^{n}\ln{(H_i\Sigma H_i^T+R_i)}-\frac{1}{2}\sum_{i=1}^{n}\frac{(Y_i-H_im)^2}{H_i\Sigma H_i^T+R_i}.
\end{equation}
The first derivatives with respect to $m_j$ and $\sigma_j^2$ are equal respectively for $1\leq j\leq p$ to$:$
\begin{equation}
\frac{\partial}{\partial m_j}l(Y|\theta)=\sum_{i=1}^{n}\frac{H_{ij}(Y_i-H_im)}{H_i\Sigma H_i^T+R_i}
\end{equation}
and 
\begin{equation}
\frac{\partial}{\partial \sigma^2_j}l(Y|\theta)=-\frac{1}{2}\sum_{i=1}^{n}\frac{H^2_{ij}}{(H_i\Sigma H_i^T+R_i)}-\frac{1}{2}\sum_{i=1}^{n}\frac{-H_{ij}^2(Y_i-H_im)^2}{(H_i\Sigma H_i^T+R_i)^2}.
\end{equation}
Then, we can calculate the second derivatives as
\begin{equation}
\label{mm}
\frac{\partial^2}{\partial m_j\partial m_k}l(Y|\theta)=-\sum_{i=1}^{n}\frac{H_{ij}H_{ik}}{H_i\Sigma H_i^T+R_i}
\end{equation}
and 
\begin{equation}
\label{sigma2sigma2}
\frac{\partial^2}{\partial\sigma^2_j \partial\sigma^2_k}l(Y|\theta)=\frac{1}{2}\sum_{i=1}^{n}\frac{H^2_{ij}H^2_{ik}}{(H_i\Sigma H_i^T+R_i)^2}-\sum_{i=1}^{n}\frac{H^2_{ij}H^2_{ik}(Y_i-H_im)^2}{(H_i\Sigma H_i^T+R_i)^3}.
\end{equation}
The cross derivatives between $m_k$ and $\sigma_j^2$ are equal to
\begin{equation}
\label{sigma2m}
\frac{\partial^2}{\partial \sigma^2_j\partial m_k}l(Y|\theta)=\sum_{i=1}^{n}\frac{-H_{ik}H_{ij}^2(Y_i-H_im)}{(H_i\Sigma H_i^T+R_i)^2}
\end{equation}
and
\begin{equation}
\label{msigma2}
\frac{\partial^2}{\partial m_k \partial \sigma^2_j}l(Y|\theta)=\sum_{i=1}^{n}\frac{-H^2_{ij}H_{ik}(Y_i-H_im)}{(H_i\Sigma H_i^T+R_i)^2}.
\end{equation}
As $\ee[Y_i-H_im]=0$ and $\ee[(Y_i-H_im)^2]=H_i\Sigma H_i^T+R_i$, the exact expression of every entry of the Fisher information matrix can be readily obtained$:$
\begin{equation}
\label{IFmoy}
I(m_j,m_k)=\sum_{i=1}^{n}\frac{H_{ij}H_{ik}}{H_i\Sigma H_i^T+R_i}
\end{equation}
\begin{equation}
\label{IFvar}
I(\sigma^2_j,\sigma^2_k)=\frac{1}{2}\sum_{i=1}^{n}\frac{H^2_{ij}H^2_{ik}}{(H_i\Sigma H_i^T+R_i)^2}
\end{equation}
and 
\begin{equation}
I(\sigma^2_j,m_k)=0\,\,\,\,\text{and}\,\,\,I(m_k,\sigma_j^2)=0.
\end{equation}

\section{Proof of Proposition \ref{ECMEbygroups}}

\label{appendixC}

The proof is similar to that of Proposition \ref{ECME}. We now have $q$ groups composing $Y$ where each of them has its own variance parameter, such that
\begin{equation}
Y:=(Y^1,\cdots,Y^s,\cdots,Y^q)^{T}\in\rr^{n}\,\,\,\,\,;\,\,\,\,\,1\leq s\leq q.
\end{equation}
with $Y^s$ the subset of $Y$ consisting of the experimental data of the $s$th group.
The way of indexing the $n$ components of $Y$ is explained in Section \ref{sec:ECME_extension} before the proposition statement.
Let us recall that$:$ 
\begin{itemize}
\item $\sigma_s^2=(\sigma^2_{s,1},\cdots,\sigma^2_{s,p})^{T}\in\rr^{p}$ is the vector of variance parameters shared by the unobserved realizations $\{\ms{\lambda}_i\}_{i_{s-1}+1\leq i\leq i_s}$ of the $s$th group,
\item $\Sigma_s$ is the diagonal matrix formed by $\sigma_s^2\in\rr^{p}$.
\end{itemize}
It follows that
\begin{equation}
\ms{\lambda}_i\thicksim\mathcal{N}(m,\Sigma_s)\,\,\,\,\,\,\,i_{s-1}+1\leq i\leq i_s.
\end{equation}
Also let us define $\ms{\delta}_i=\ms{\lambda}_i-m$. If $i_{s-1}<i\leq i_s$,
\begin{equation}
Y_i=H_i(\ms{\delta}_i+m)+\epsilon_i
\end{equation}
where
\begin{equation}
\ms{\delta}_i\thicksim\mathcal{N}(0,\Sigma_s),
\end{equation}
with $\ms{\delta}_i=(\delta_{i1},\cdots,\delta_{ip})^{T}\in\rr^{p}$. If $Y_i$ belongs to the $s$th class, we have for $1\leq j\leq p$
\begin{equation}
\ee\left( \begin{array}{c}
     Y_i \\
     \delta_{ij}
\end{array} \right)=\left(\begin{array}{c}
H_{i}m \\
0
\end{array}\right)
\end{equation}
and 
\begin{equation}
\vv\left( \begin{array}{c}
     Y_i \\
     \delta_{ij}
\end{array}  \right)=\left(\begin{array}{cc}
H_{i} \diag{\sigma^2_s} H_{i}^{T}+R_i & (H_{ij}) \sigma^2_{s,j} \\
\sigma^2_{s,j} H_{ij} & \sigma^2_{s,j}
\end{array}\right).
\end{equation}
Let us define
\begin{equation}
A_{i}=Y_{i}-H_{i}m ; \quad B_{ij}=\sigma^2_{s,j} H_{ij}; \quad V_{i}=H_{i} \Sigma_s H_{i}^{T}+R_i.
\end{equation}
With these notations, we have
\begin{align}
\nonumber
    \ee\left(\delta_{ij}\mid Y_{i}\right)=&  \sigma^2_{s,j} H_{ij}(H_i \diag{\sigma^2_s} H_i^T + R_i)^{-1}(Y_i - H_i m) \\
    =& B_{ij} V_{i}^{-1} A_{i}
\label{espcond_delta}
\end{align}
and
\begin{align}
\nonumber
\vv \left( \delta_{ij}  \mid Y_{i}\right)=& \sigma^2_{s,j} - (\sigma^2_{s,j} H_{ij}) (H_i \diag{\sigma^2_s} H_i^T + R_i)^{-1} (H_{ij} \sigma^2_{s,j})\\
=&\sigma^2_{s,j}-  B_{ij}^2 V_i^{-1}.
\label{varcond_delta}
\end{align}
Let $Z_i=(Y_i,\ms{\delta}_i)\in\rr^{p+1}$ be the union of the $i$th experimental data and unobserved model realization shifted by $m$.
The E. step of the ECME algorithm is to derive the expectation of the complete log-likelihood with respect to the distribution of $\delta$ conditional on $(Y,\theta^k)$, equal to
\begin{equation}
\label{Qfunction}
Q(\theta,\theta^k)=\ee_{\delta}[l(Z|\theta)|Y,\theta^k]
\end{equation}
with $\theta^k=(m^k,(\sigma_1^2)^k,\cdots,\cdots,(\sigma_q^2)^k)$ and $\delta$ synthesizing the unobserved samples $\{\ms{\delta}_i\}_{1\leq i\leq n}$. The complete log-likelihood
is written as:
\begin{equation}
l(Z|\theta)=l(Y|\delta,\theta)+l(\delta|\theta).
\end{equation}
The first term is in fact independent of $\theta$ and thus maximizing (\ref{Qfunction}) comes down to maximizing $\ee_{\delta}[l(\delta|\theta)|Y,\theta^k]$.
Defining the following canonical parameter $\psi^{T}=(\psi^T_1,\cdots,\psi^T_q)=((\sigma_1^2)^{-1},\cdots,(\sigma_q^2)^{-1})$ leads to 
\begin{multline}
\label{logQpartial}
l(\delta|\psi)=-\frac{np}{2}\ln(2\pi) + \frac{n_1}{2}\ln(|\Sigma_1|)+\dots + \frac{n_l}{2}\ln(|\Sigma_q|) \\ 
-\frac{1}{2} \sum_{i=1}^{i_1} (\ms{\delta}_i^2)^T \psi_1 - \dots -\frac{1}{2} \sum_{i=i_{q-1} +1 }^{i_q} (\ms{\delta}_i^2)^T \psi_q,
\end{multline}
with $\ms{\delta}_i^2:=(\delta_{i1}^2,\cdots,\delta_{ip}^2)\in\rr^{p}$. We can write Eq. (\ref{logQpartial}) in the following way
\begin{equation}
l(\delta|\psi)=\psi^T t(\delta) - \ln (a(\psi)) + \ln (b(\delta)) 
\end{equation}
with
\begin{itemize}
\item $\psi^T=(\psi^T_1,\dots, \psi^T_q)$ the canonical parameter$;$

\smallskip
\item $t(\delta)^{T}=\left( - (1/2)\sum_{i=1}^{i_1}  (\ms{\delta}_i^2)^T, \dots, - (1/2)\sum_{i=i_{q-1}+1}^{i_q} (\ms{\delta}_i^2)^T  \right)$ the sufficient statistics of $\lambda$;

\smallskip
\item$\ln(a(\psi))= -\left(\frac{n_1}{2} \ln|\text{diag}(\psi_1)| + \dots + \frac{n_q}{2} \ln|\text{diag}(\psi_q)|\right)  $;

\smallskip
\item $\ln(b(\delta))=-\frac{np}{2}\ln(2\pi)$ depending on the length of $\delta$. 
\end{itemize} 
Considering the canonical parameter instead of $\theta$, we have (up to a constant)
\begin{align}
\nonumber
Q(\psi,\psi^k)= &\,\ee_{\delta}[l(\delta|\psi)|Y,\psi^k]\\
\nonumber
=& \,\ee_{\delta}(\psi^T t(\delta)|Y,\psi^k) - \ee_{\delta}(\ln(a(\psi))|Y,\psi^k) + \ee_{\delta}(\ln(b(\delta))|Y,\psi^k) \\
    =& \,\psi^T \ee_{\delta}(t(\delta)|Y,\psi^k) - \ln(a(\psi)) + \ln(b(\delta))
\label{Estep}
\end{align}
where the $j$th component of $\ee_{\delta}(t(\delta)|Y,\psi^k)$ limited to the $s$th group is equal to
\begin{equation}
 \ee\Big(-\frac{1}{2} \sum_{i=i_{s-1}+1}^{i_s} \delta_{ij}^2 \mid Y, \theta^{k}\Big)=-\frac{1}{2} \sum_{i=i_{s-1}+1}^{i_s}\left[\ee\left(\delta_{ij} \mid Y, \theta^{k}\right)^2 +\vv\left(\delta_{ij} \mid Y, \theta^{k}\right)\right].
\end{equation} 
Using Eqs. (\ref{espcond_delta}) and (\ref{varcond_delta}) we thus have
\begin{multline}
\ee\Big(-\frac{1}{2} \sum_{i=i_{s-1}+1}^{i_s} \delta_{ij}^2 \mid Y, \theta^{k}\Big) = \\
-\frac{1}{2}\Big(n_s (\sigma_{s,j}^2)^k + \sum_{i=i_{s-1}+1}^{i_s} \Big[(B_{ij}^k(V_i^k)^{-1} A_i^k)^2  - (B_{ij}^k)^{2}(V_i^k)^{-1} \Big]\Big).
\end{multline}
Then, the CM1 step is to find $\psi^{k+1}$ that maximizes (\ref{Estep}).
\begin{align*}
    &\frac{\partial Q(\psi,\psi^k)}{\partial \psi}=0\\
    \Longleftrightarrow &\,\,\,\frac{\partial \ee_{\delta}[l(\delta|\psi)|Y,\psi^k])}{\partial \psi}=0\\
    \Longleftrightarrow &\,\,\, \frac{\partial [\psi^T \ee_{\delta}(t(\delta)|Y,\psi^k) - \ln(a(\psi)) + \ln(b(\delta))]}{\partial \psi}=0\\
    \Longleftrightarrow & \,\,\,\ee_{\delta}(t(\delta)|Y,\psi^k) - \frac{\ln(a(\psi))}{\partial \psi}=0\\
    \Longleftrightarrow &\,\,\, \ee_{\delta}(t(\delta)|Y,\psi^k) = \frac{\partial \ln(a)}{\partial \psi}(\psi).
\end{align*}
It comes out that $\psi_{k+1}$ achieves the following equation
\begin{equation}
\ee_{\delta}(t(\delta)|Y,\psi^k) = \frac{\partial \ln(a)}{\partial \psi}(\psi_{k+1}).
\end{equation}
For $1\leq s\leq q$, 
\begin{align}
    \frac{\partial \ln(a)}{\partial \psi_s}(\psi)=&-\frac{n_s}{2} \psi_s^{-1}\\
    =& -\frac{n_s}{2} \sigma^2_s.
\end{align}
Finally, we can deduce the expression of each $\sigma_s^2=\psi_s^{-1}$ at the $k+1$th iteration of this ECME algorithm. For $1\leq s\leq q$$:$
\begin{equation}
(\sigma^2_{s,j})^{k+1}=(\sigma_{s,j}^2)^k + \frac{1}{n_s}\,\sum_{i={i_s-1}+1}^{i_s} \left( \left(B_{ij}^k (V_i^k)^{-1}A_i^k\right)^2 - (B_{ij}^k)^2 (V_i^k)^{-1}\right).
\end{equation}
The CM2 step is then to compute $m_{k+1}$ that maximizes the incomplete likelihood conditional on all $(\sigma^2_s)^{k+1}_j$. It is given by the weighted least square estimator below
\begin{equation}
m^{k+1}= \left( \sum_{i=1}^n H_i^T (V_i^{k+1})^{-1} H_i\right)^{-1} \left( \sum_{i=1}^n H_i^T (V_i^{k+1})^{-1} Y_i \right),
\end{equation}
with $V_i^{k+1}$ being equal to $H_i\text{diag}((\sigma^2_s)^{k+1})H_i^{T}+R_i$ for $i_{s-1}+1 \leq i\leq i_s$.

\end{document}